\documentclass[prodmode,acmtecs]{acmsmall} % Aptara syntax

\usepackage[ruled]{algorithm2e}

\SetAlFnt{\small}
\SetAlCapFnt{\small}
\SetAlCapNameFnt{\small}
\SetAlCapHSkip{0pt}
\IncMargin{-\parindent}

\usepackage{amsmath}
\usepackage[table,xcdraw]{xcolor}
\usepackage{array}
\newcolumntype{x}[1]{>{\centering\let\newline\\\arraybackslash\hspace{0pt}}m{#1}}
\usepackage{subfigure}
\usepackage{makecell}
\usepackage{mathtools}

\newcommand{\ignore}[1]{}

\acmVolume{9}
\acmNumber{4}
\acmArticle{39}
\acmYear{2010}
\acmMonth{3}

\doi{0000001.0000001}

\issn{1234-56789}

\begin{document}

% Page heads
\markboth{H. Jung et al.}{Multiprocessor Scheduling of a Multi-mode Dataflow Graph Considering Mode Transition Delay}

% Title portion
\title{Multiprocessor Scheduling of a Multi-mode Dataflow Graph Considering Mode Transition Delay}
\author{HANWOONG JUNG
\affil{Seoul National University}
HYUNOK OH
\affil{Hanyang University}
SOONHOI HA
\affil{Seoul National University}}
% NOTE! Affiliations placed here should be for the institution where the
%       BULK of the research was done. If the author has gone to a new
%       institution, before publication, the (above) affiliation should NOT be changed.
%       The authors 'current' address may be given in the "Author's addresses:" block (below).
%       So for example, Mr. Abdelzaher, the bulk of the research was done at UIUC, and he is
%       currently affiliated with NASA.

\begin{abstract}
Synchronous Data Flow (SDF) model is widely used for specifying signal processing or streaming applications. Since modern embedded applications become more complex with dynamic behavior changes at run-time, several extensions of the SDF model have been proposed to specify the dynamic behavior changes while preserving static analyzability of the SDF model. They assume that an application has a finite number of behaviors (or modes) and each behavior (mode) is represented by an SDF graph. They are classified as multi-mode dataflow models in this paper. While there exist several scheduling techniques for multi-mode dataflow models, no one allows task migration between modes. By observing that the resource requirement can be additionally reduced if task migration is allowed, we propose a multiprocessor scheduling technique of a multi-mode dataflow graph considering task migration between modes. Based on a genetic algorithm, the proposed technique schedules all SDF graphs in all modes simultaneously to minimize the resource requirement. To satisfy the throughput constraint, the proposed technique calculates the actual throughput requirement of each mode and the output buffer size for tolerating throughput jitter. We compare the proposed technique with a method which analyzes SDF graphs in each execution mode separately and a method that does not allow task migration for synthetic examples and three real applications: H.264 decoder, vocoder, and LTE receiver algorithms.
\end{abstract}

\begin{CCSXML}
<ccs2012>
  <concept>
    <concept_id>10003752.10003753.10003760</concept_id>
    <concept_desc>Theory of computation~Streaming models</concept_desc>
    <concept_significance>300</concept_significance>
  </concept>
  <concept>
    <concept_id>10010520.10010553.10010562.10010564</concept_id>
    <concept_desc>Computer systems organization~Embedded software</concept_desc>
    <concept_significance>300</concept_significance>
  </concept>
</ccs2012>
\end{CCSXML}

\ccsdesc[300]{Theory of computation~Streaming models}
\ccsdesc[300]{Computer systems organization~Embedded software}


% We no longer use \terms command
%\terms{Design, Algorithms, Performance}

\keywords{Synchronous dataflow, Multi-mode dataflow, Mode transition delay, Task migration, Throughput requirement}

\acmformat{Hanwoong Jung, Hyunok Oh, and Soonhoi Ha, 2016. Multiprocessor Scheduling of a Multi-mode Dataflow Graph Considering Mode Transition Delay.}
% At a minimum you need to supply the author names, year and a title.
% IMPORTANT:
% Full first names whenever they are known, surname last, followed by a period.
% In the case of two authors, 'and' is placed between them.
% In the case of three or more authors, the serial comma is used, that is, all author names
% except the last one but including the penultimate author's name are followed by a comma,
% and then 'and' is placed before the final author's name.
% If only first and middle initials are known, then each initial
% is followed by a period and they are separated by a space.
% The remaining information (journal title, volume, article number, date, etc.) is 'auto-generated'.

\begin{bottomstuff}
This research was supported by a grant to Bio-Mimetic Robot Research Center Funded by Defense Acquisition Program Administration, and by Agency for Defense Development (UD130070ID), Basic Science Research Program through the National Research Foundation of Korea (NRF) funded by the Ministry of Science, ICT \& Future Planning (NRF-2013R1A2A2A01067907, 2013R1A1A1013384), and IT R\&D program MKE/KEIT (No. 10041608, Embedded system Software for New-memory based Smart Device).

Author's addresses: H. Jung {and} S. Ha, Department of Computer Science and Engineering,
Seoul National University; H. Oh, Department of Information System, Hanyang University; 
\end{bottomstuff}

\maketitle

\section{Introduction}
\label{Section:Introduction}

Model-based design methodology is widely accepted for embedded system design since it enables us to cope with ever increasing system complexity by maximizing the benefit of abstraction. As an algorithm specification model, this paper adopts a coarse-grain dataflow model which is suitable for specifying signal processing or streaming applications. In a dataflow graph, a node presents a function module and an arc represents the flow of data samples (or tokens) through the FIFO channel between two end nodes. When a node is invoked, it consumes a specified number of data samples (called sample rate) from each input arc and produces a specified number of samples to each output arc. A node becomes executable when all input arcs have as many data samples as the specified sample rates. If the sample rate is a fixed integer number which does not change at run-time, the dataflow graph is called a synchronous dataflow (SDF) graph \cite{Lee:1987}.

For a given multiprocessor system, we need to determine the mapping of nodes to the processors and the execution order of mapped nodes on each processor. The static sample rate in the SDF model allows us to make the mapping and scheduling decision statically. From the static mapping and scheduling result for an SDF graph on a multiprocessor, we can estimate the performance and the resource requirement, which is very desirable for the design of embedded systems with tight real-time and resource constraints. If the implemented system follows the pre-determined mapping and execution order of nodes at run-time, the system can be claimed to be ``\textit{correct by construction}''.

But the SDF model has a severe restriction to be used for modern embedded applications. It cannot express the dynamic behavior of an application, while modern embedded applications become more complex with dynamic behavior changes at run-time. For example, advanced video CODEC algorithms have several function modules that are conditionally invoked depending on the contents of the input frame. In addition, an application may have multiple implementations of the same algorithm to support various levels of quality of service.

To express such dynamic behavior of an application in the SDF model with keeping the static analysis capability, several extensions have been proposed to the SDF model, including FSM-based scenario-aware dataflow (FSM-SADF) \cite{Stuijk:2011}, parameterized SDF (PSDF) \cite{Bhattacharya:2001}, MTM-SADF \cite{Jung:2014}, mode-aware dataflow (MADF) \cite{Zhai:2015}, and so on \cite{Girault:1999} \cite{Wiggers:2008}. They all assume that an application has a finite number of behaviors (or modes) and each behavior (mode) can be represented by an SDF graph. We denote those MoCs (Models of Computation) as multi-mode dataflow (MMDF) graphs in this paper and define a representative MMDF model that can be implemented by any specific extension.

Because an MMDF graph is composed of SDF graphs in multiple modes, the static schedule of the SDF graph on each mode can be constructed to estimate the overall performance and the resource requirement of the MMDF graph. Also, for more accurate estimation, the mode transition delay during mode changes should be considered. If the mode change occurs frequently and periodically, the mode transition delay will seriously degrade the overall performance of the MMDF graph. While there exist several techniques to schedule a multi-mode dataflow graph with considering the mode transition delay, the existing approaches do not allow task migration between modes. However, we observe that the resource requirement for a given throughput constraint can be additionally reduced if task migration is allowed. Since task migration will cause additional run-time overhead during mode transitions, we should take into account the effect task migration overhead on the throughput performance.

In this paper, we propose a multiprocessor scheduling technique of a multi-mode dataflow graph considering the mode transition delay conservatively. Our scheduling objective is to minimize the number of processors allowing task migration among all modes, while satisfying the overall throughput constraint. Also, from the scheduling result, we compute the output buffer size to tolerate the time fluctuation of output results. Experiment results show that the proposed scheduling approach provides better solutions than the existing approaches which determine the mapping and the scheduling for graphs either each mode independently incurring task migration between modes, or all modes at the same time disallowing task migration.

The rest of this paper is organized as follows. The next section gives a motivational example to clarify the problem addressed in this paper, and introduces the key idea of the proposed technique. Section \ref{Section:Related Work} reviews the related work. The problem addressed in this paper will be defined and formulated in Section \ref{Section:Problem Definition}. In Section \ref{Section:Throughput Requirement Analysis} and \ref{Section:Proposed MMDF Scheduling Framework}, the throughput requirement analysis technique and the proposed scheduling technique considering the mode transition delay are explained in detail, respectively. In Section \ref{Section:Experimental Result}, we discuss our experimental results, and draw conclusions in Section \ref{Section:Conclusion}.

\section{Motivational Example}
\label{Section:Motivational Example}

\subsection{Throughput Requirement Calculation Considering Mode Transition Delay}
\label{SubSection:Mode Transition Delay}

\begin{figure}[ht]
\centering
  \begin{subfigure}{}
    \includegraphics[width=0.8\columnwidth,keepaspectratio]{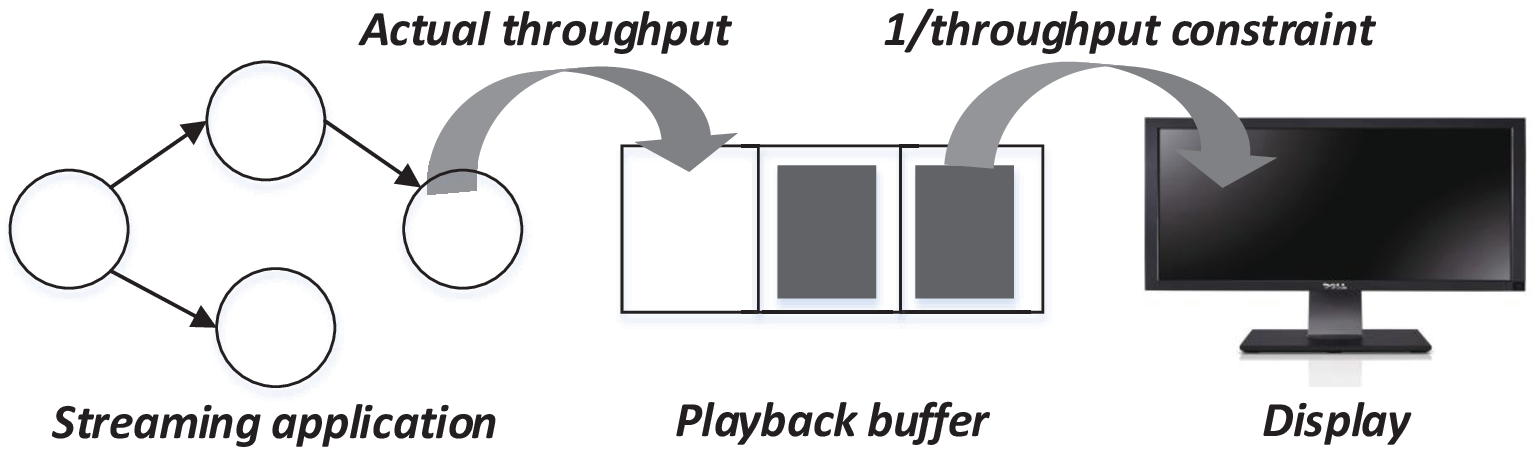}
    \\
    \subcaption{(a) Relation between the playback buffer and throughput}
  \end{subfigure}
  \hfill
  \\
  \begin{subfigure}{}
    \includegraphics[width=0.9\columnwidth,keepaspectratio]{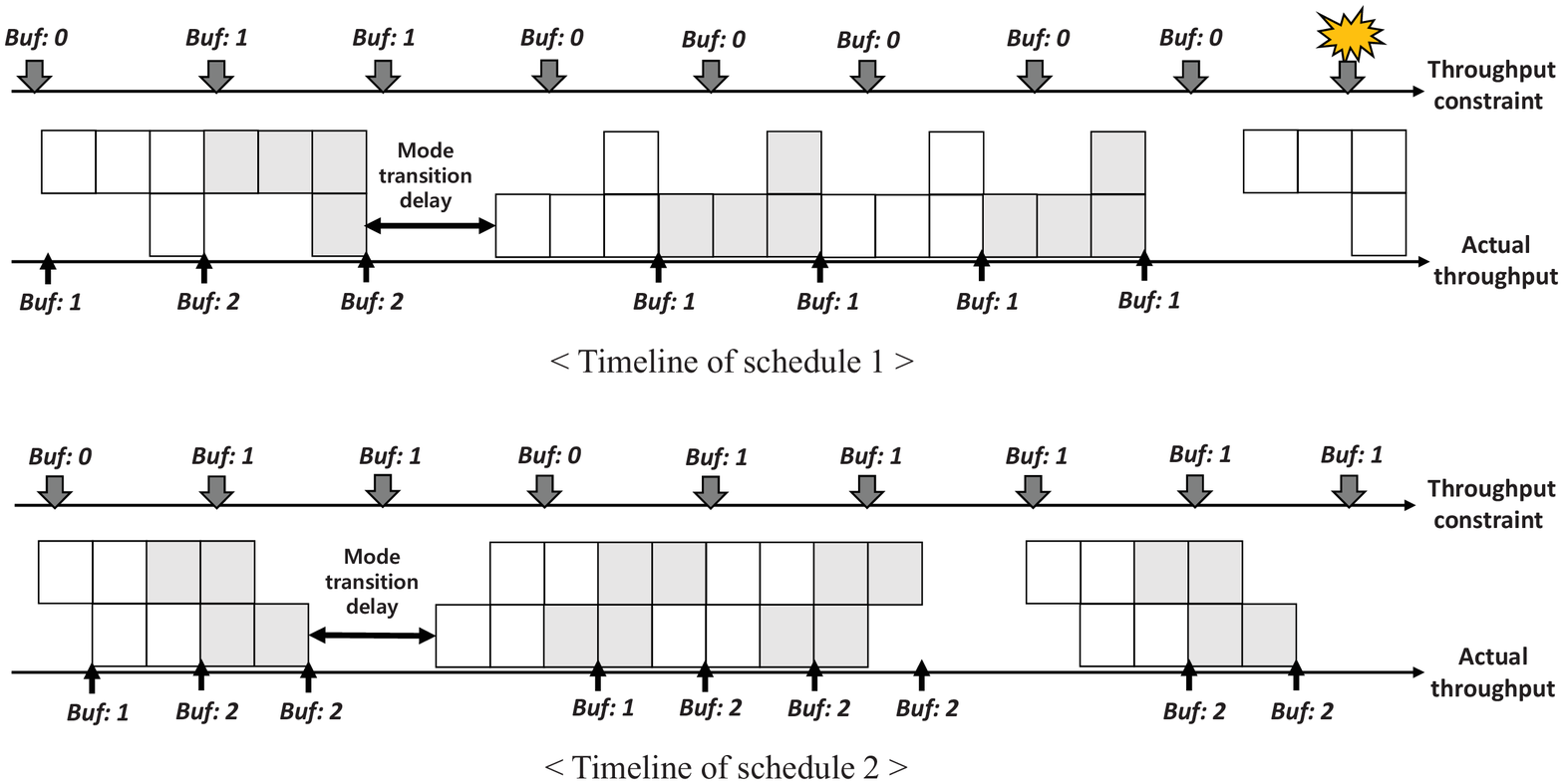}
    \\
    \subcaption{(b) Throughput requirement considering the mode transition delay}
  \end{subfigure}
  \caption{Motivational example of throughput requirement calculation considering the mode transition delay}
  \label{Figure:Motivational Example of Mode transition}
\end{figure}

To cope with fluctuation of output intervals of streaming applications, an output buffer is usually used to obtain the steady output stream in streaming applications as depicted in Figure \ref{Figure:Motivational Example of Mode transition} (a). The system will require data from the output buffer periodically and the period is defined by the inverse of the throughput constraint. If the throughput performance of an application is lower than the throughput constraint, the output buffer will be eventually empty.

If a streaming application is specified in an MMDF graph, it can be scheduled at compile-time to meet the given throughput constraint. To guarantee the throughput constraint, not only throughput performance of each mode, but also the mode transition delay should be considered. Although a mapping/scheduling result of each SDF graph keeps the throughput constraint, the average throughput performance can be lower than the constraint because of additional time delay during mode transition.

Various factors have an effect on the mode transition delay. In \cite{Stuijk:2010} and \cite{Geilen:2012}, the system reconfiguration overhead and DVFS delay are considered as the mode transition delay. Also, \cite{Zhai:2015} defines the mode transition delay to quantify the proposed transition protocols. Especially, it proposes a MOO (Maximum-Overlap Offset) transition protocol which calculates an offset that guarantees no interference between the execution of both old and new modes. When a mode transition occurs, an application will be delayed until this offset, which should be considered in the throughput calculation. In addition to such rescheduling delay, we need to consider the task migration overhead in the computation of mode transition delay.

Since the mode transition delay degrades the average throughput of an MMDF application, the throughput constraint can be violated even if the throughput of each mode is higher than the throughput constraint. Figure \ref{Figure:Motivational Example of Mode transition} (b) shows an example of two different schedules for an MMDF graph. The MMDF graph consists of two different modes, and there exists additional time delay during mode transition. As represented by arrows on an upper line, the system dequeues data from the output buffer periodically with the same rate as the throughput constraint. The arrows on a lower line tell when an MMDF application enqueues data to the output buffer. A number annotated on each arrow denotes the number of data in the output buffer after the access is completed.

In case of schedule 1 in Figure \ref{Figure:Motivational Example of Mode transition} (b), even though the schedule of each mode satisfies the throughput constraints, the throughput constraint is eventually violated since the mode transition delay is accumulated. To avoid this problem, we need to set the throughput constraint of each mode tighter than that of the application as schedule 2 illustrates in the figure; it keeps the throughput constraint because it fills the output buffer faster than the throughput constraint. Therefore we need to calculate the actual throughput requirement for each mode considering the mode transition delay, in order not to violate the given constraint. Details will be discussed in Section \ref{Section:Throughput Requirement Analysis}.

\subsection{Task Migration Between Mode Transition}
\label{SubSection:Task migration}

\begin{figure} [ht]
\centerline{\includegraphics[width=0.95\textwidth]{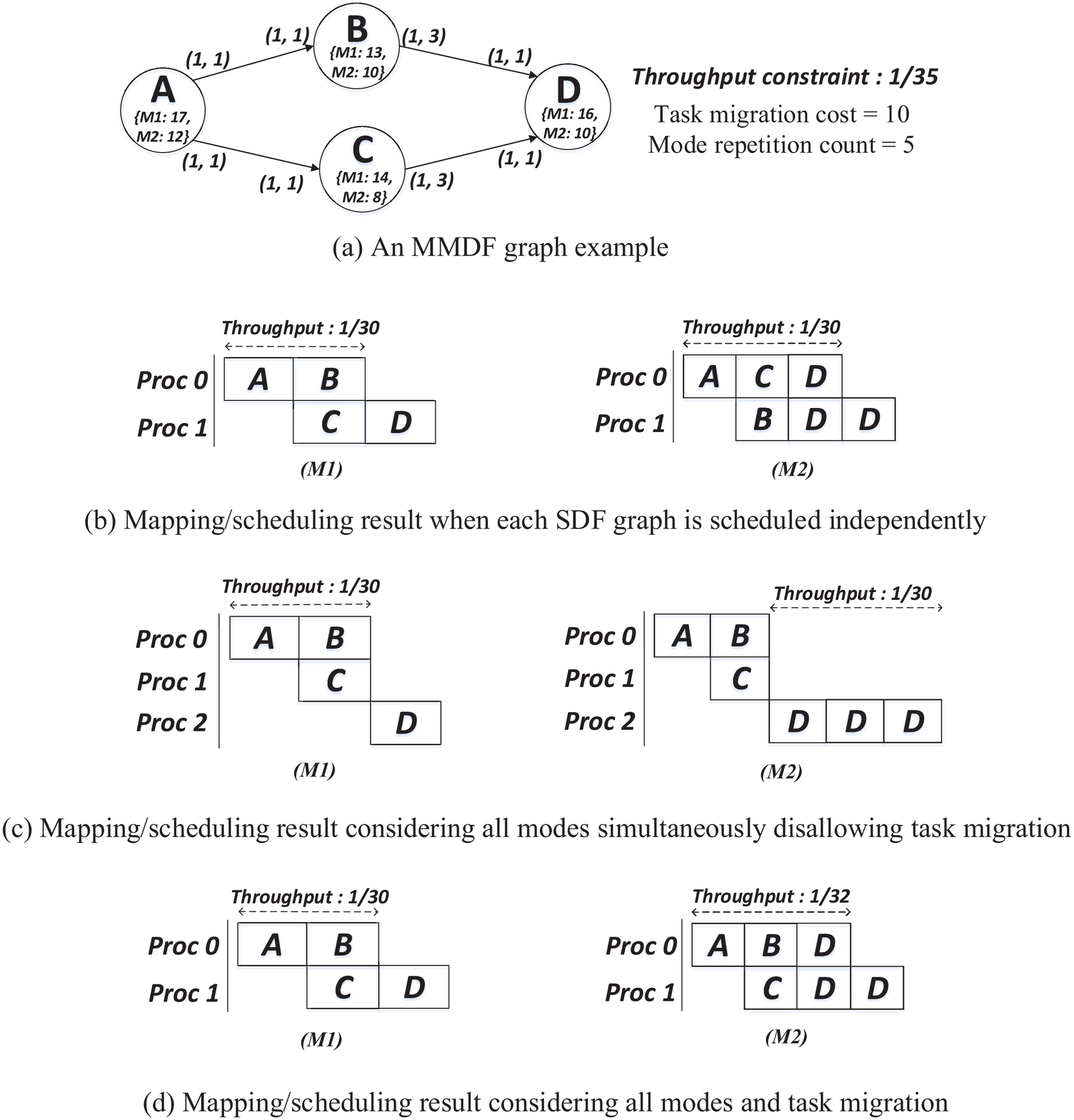}}
\caption{Motivational Example of Task Migration}
\label{Figure:Motivational Example of Task Migration}
\end{figure}

Figure \ref{Figure:Motivational Example of Task Migration} (a) shows an MMDF graph example which consists of two modes: \textit{M1} and \textit{M2}. We assume that the throughput constraint of the MMDF graph is given as 1/35. Each execution mode of an MMDF graph is represented with an SDF graph. The execution time and sample rates of each node may vary depending on the execution mode. In mode \textit{M1}, the execution times of nodes \textit{A, B, C,} and \textit{D} are 17, 13, 14, and 16, respectively and in mode \textit{M2}, they are 12, 10, 8, and 10. The output sample rates of nodes \textit{B} and \textit{C} are unity in mode \textit{M1} while they are 3 in mode \textit{M2}. Refer to Section \ref{Section:Problem Definition} for the formal description of the MMDF model assumed in this paper. Also, in this section, we only consider the task migration overhead as the mode transition delay to simply show the effect of task migration during mode transition.

 A naive approach to schedule an MMDF graph is to schedule an SDF graph in each mode independently with multiple objectives of resource minimization and throughput maximization. For example, for the given throughput constraint, we find an optimal mapping/scheduling result in each execution mode as shown in Figure \ref{Figure:Motivational Example of Task Migration} (b). Since it does not consider mapping results in the other modes, a node may be mapped onto different processors between modes. Therefore, the mapping result requires task migration when the mode changes. In Figure \ref{Figure:Motivational Example of Task Migration} (b), nodes \textit{B}, \textit{C} and \textit{D} will be migrated to other processors when the mode transition occurs.

Another approach to schedule an MMDF graph is to consider all modes simultaneously disallowing task migration \cite{Stuijk:2008} \cite{Geilen:2010}. Since the mapping is constrained in these approaches, the scheduling results generally require more processors than those that allow task migration. For instance, three processors are required to meet the given throughput constraint for the mapping/scheduling result without task migration as shown in Figure \ref{Figure:Motivational Example of Task Migration} (c), while two processors are enough for the scheduling result with task migration in Figure \ref{Figure:Motivational Example of Task Migration} (b). Since the objective of this paper is to minimize the resource requirement under a given throughput constraint, the proposed approach allows task migration. Their approach is used as a reference technique for comparison with the proposed technique in experiments.

Consider the former approach that allows task migration in Figure \ref{Figure:Motivational Example of Task Migration} (b). If the mode transition occurs frequently and the task migration overhead is non-negligible, then the given throughput constraint may not be satisfied. For instance, assume that the mode transition occurs every 5 iterations and the task migration overhead of each node is 10. In Figure \ref{Figure:Motivational Example of Task Migration} (b), 30 time unit is added every 5 iterations because nodes \textit{B}, \textit{C} and \textit{D} should be migrated for mode transition. Then, the output buffer will be eventually empty, because the average throughput performance of the MMDF graph becomes lower then the throughput constraint.

Therefore, in this paper, we propose another approach that schedules the SDF graphs of all modes simultaneously allowing task migration among execution modes. Figure \ref{Figure:Motivational Example of Task Migration} (d) shows a mapping and scheduling result produced by the proposed technique. It requires 2 processors and only 10 additional time units for task migration, which may satisfy the throughput requirement with proper output buffering. Throughput analysis considering task migration overhead will be discussed in Section \ref{Section:Throughput Requirement Analysis}.

\section{Related Work}
\label{Section:Related Work}

As the related work to the proposed technique, we review several extensions of the SDF model that have been proposed to express the dynamic behavior of an application.

One of the most representative multi-mode dataflow models is FSM-based SADF (Scenario-Aware Data Flow) model \cite{Stuijk:2008}, shortly FSM-SADF. In the FSM-SADF model, an application consists of multiple scenarios (modes) and each scenario is specified by an SDF graph. To specify multiple scenarios and their transitions, it defines a special control task called detector that has an FSM inside. The detector task sends the control information to the normal computation tasks that may change its behavior. For the FSM-SADF model, several techniques to statically analyze the timing behavior such as worst case latency and throughput \cite{Geilen:2010} have been proposed. Also, a binding-aware scenario graph \cite{Stuijk:2010} has been proposed to take into account the resource constraint. And, in \cite{Damavandpeyma:2013}, it considers reconfiguration overhead for DVFS (Dynamic Voltage Frequency Scaling) as the mode transition delay. However, it only considers the worst-case performance analysis of the FSM-SADF graph for the given task mapping, and requires inherently exponential time-complexity for exact analysis.

As a similar model to the FSM-SADF, an MTM-SADF \cite{Jung:2014} has been proposed to specify application level dynamism based on an SDF. Instead of an FSM, it uses a Mode Transition Machine (MTM) which is a simplified form of the FSM to represent the mode transition. It proposes a hybrid task mapping technique with minimizing the overall energy consumption under the throughput constraints. However, it analyzes each SDF graph independently.

PSDF (Parameterized Synchronous Data Flow) \cite{Bhattacharya:2001} proposes a meta-modeling technique for run-time adaptation of parameters in a structured way. In the PSDF model, the dynamic behavior of a task is modeled by parameters and the task behavior can change at the iteration boundary at run-time. Since the PSDF becomes an SDF graph at each iteration, the PSDF can be regarded as a multi-mode dataflow graph that may change modes every iteration.

MCDF (Mode-Controlled Data Flow) \cite{Moreira:2012} is one of data flow MoCs which enables the expression of the data-dependent functional behavior. However, it mainly focuses on SDR (Software-Defined Radio) applications, where different sub-graphs need to be active in different modes.

In VRDF (Variable-Rate Data Flow) \cite{Wiggers:2008} model, it allows variable port rates within a specified range, and VPDF (Variable-rate Phase Data Flow) \cite{Wiggers:2008} is proposed to combine characteristics of VRDF and CSDF where each actor has a sequence of phases, and for every phase, the number of firings can be parameterized. For these MoCs, buffer size analysis technique is proposed to satisfy a throughput constraint.

MADF (Mode-Aware Data Flow) \cite{Zhai:2015} has been proposed to support hard real-time scheduling for multi-mode CSDF (Cyclo-Static Data Flow) model \cite{Bilsen:1995}. It combines advantages of SADF and VPDF to specify application level dynamism. Also, it proposes MOO (Maximum-Overlap Offset) mode transition protocol to derive an efficient analysis for a hard real-time scheduling of an MADF graph. With this mode transition protocol, the timing behavior of individual modes and during mode transitions can be analyzed independently.

BPDF (Boolean Parametric Data Flow) \cite{Bebelis:2013} supports change of port rates and graph topology at run-time using integer and boolean parameters. In BPDF model, integer parameters are used to change port rates at each iteration, and boolean parameters are used for activation and deactivation of edges to change graph topology.

HDF (Heterogeneous Data Flow) (or *-chart) \cite{Girault:1999} supports multi-mode applications through an FSM that executes an iteration of an SDF graph in each state. So, an application is specified with a set of different SDF graphs combined with an FSM.

While various analysis and scheduling techniques have been proposed for those MoCs, no one considers task migration between modes. In \cite{Lee:2013}, task migration is considered in the failure-aware task scheduling technique where an SDF graph is scheduled multiple times with different number of processors allocated, aiming to maximize the throughput with the allocated number of processors. When a processor fails in the middle of execution, it changes the schedule that uses the reduced number of processors by one. Then task migration occurs between two different schedules before and after processor failure. They try to minimize the migration cost between two SDF schedules. This method is similar to the base method that will be used for comparison in this paper: schedule each mode separately and find the best processor-to-processor mapping (or processor renaming) in order to minimize the migration cost.

In summary, to the best of our knowledge, this paper is the first work which proposes a multiprocessor scheduling technique of an MMDF graph allowing task migration between modes, and analyzes the throughput requirement considering the mode transition delay.

\section{Problem Definition}
\label{Section:Problem Definition}

\begin{figure} [ht]
\centerline{\includegraphics[width=0.7\textwidth]{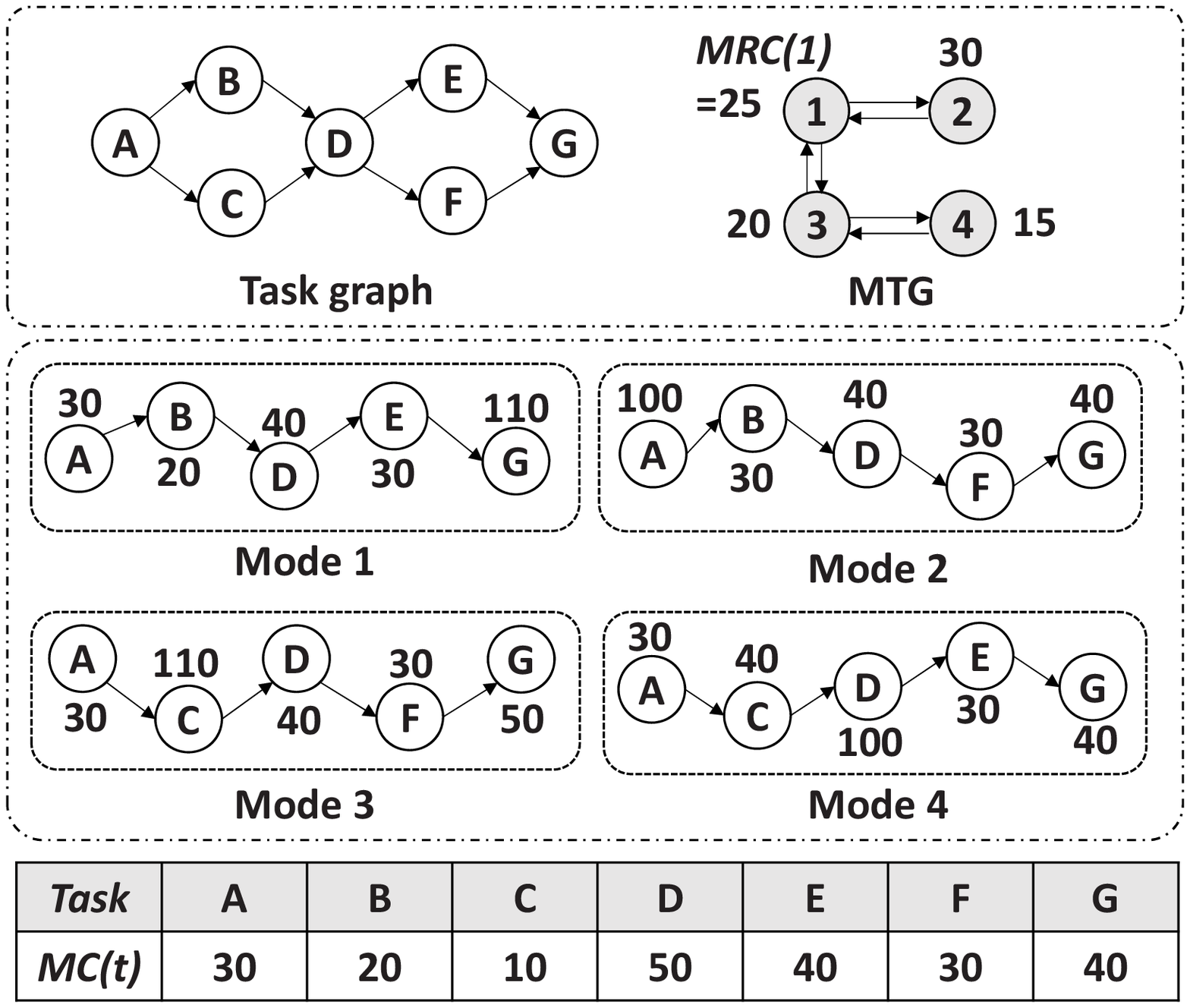}}
\caption{An MMDF graph example}
\label{Figure:MMDF Graph Example}
\end{figure}

The MMDF model assumed in this paper is not a specific model but a generic model encompassing existing similar models such as FSM-SADF \cite{Stuijk:2008} and MTM-SADF \cite{Jung:2014}. In those models, mode transition is specified by an FSM and all modes are integrated into a single graph with varying configuration parameters. Figure \ref{Figure:MMDF Graph Example} shows an MMDF graph example. We first define the MMDF model and the problem formally.

\vspace{0.3cm}

\textbf{\underline{Application model}}: An MMDF graph is specified by a combination of a task graph and a mode transition graph ($MTG$), or $(T, C, D) \times MTG$, where

\begin{itemize}
  \item $MTG$ is specified by a tuple ($Mode$, $Trans$) where $Mode$ is a finite set of modes and $Trans$ is a finite set of transitions. $Trans$ is specified as follows: $Trans =\lbrace (cm,nm) | cm \in Mode, nm \in Mode \rbrace$ where $cm$ denotes a current mode and $nm$ denotes a next mode.
  \vspace{0.2cm}
  \item $T$ is a finite set of computational tasks. Each task $t \in T$ has a set of ports $P_t$ to send/receive data to/from other adjacent tasks. $P_t = IP_t \cup OP_t$ where $IP_t$ is a set of input ports and $OP_t$ is a set of output ports. For each port $p \in P_t$, it is assigned a fixed rate, $Rate(p, mode)$, in each execution mode. Thus the graph becomes an SDF graph for each mode.
  \vspace{0.2cm}
  \item $C$ is a finite set of FIFO channels. A channel defines a one-to-one connection between two end ports.
  \vspace{0.2cm}
  \item $D$ is a set of the number of initial tokens in all channels. $d_c (m) \in \lbrace 0 \rbrace \cup \mathbb{N}$ for $\forall d_c \in D$ is the number of initially stored tokens in the channel $c$ in mode $m$.
\end{itemize}

\textbf{\underline{Architecture model}}: the target architecture consists of a set of processing elements.

\begin{itemize}
  \item $PE$ is a set of processing elements. For each $p \in PE$ and $m \in Mode$, $Map(m, p) = \lbrace t | t \in T$ where $t$ is mapped onto a processor $p$ in mode $m \rbrace$
\end{itemize}

Note that even though the proposed technique is applicable to heterogeneous multiprocessor systems, this paper assumes a homogeneous multiprocessor system for simple explanation and implementation.

To analyze the scheduling performance of an MMDF graph, we assume profiling information is available as follows:

\vspace{0.3cm}

\textbf{\underline{Profiling information}}

\begin{itemize}
  \item Worst case execution time ($WCET$) for each task $t \in T$ and $m \in Mode$ is given as $WCET(t_m, p)$ for each processing element $p \in PE$ of the target architecture. In Figure \ref{Figure:MMDF Graph Example}, the $WCET$ of a node is annotated in each mode. For example, $WCET$ of node \textit{A} is 30 in modes 1, 3, and 4, and 100 in mode 2.
  \vspace{0.2cm}
  \item For each $m \in Mode$, we are given a minimum number of iterations that the application stays at the mode, which is denoted by $MRC(m)$ where $MRC$ stands for the minimum repetition count. As $MRC$ becomes smaller, the mode transition occurs more frequently. A mode is associated with an $MRC$ value as shown in Figure \ref{Figure:MMDF Graph Example} where $MRC$ is 25, 30, 20 and 15 in modes 1, 2, 3, and 4, respectively.
  \vspace{0.2cm}
  \item For each $t \in T$, task migration cost is given by $MC(t)$. If the system is a distributed memory system, the migration cost will include the time overhead of moving the code and the context of a task between two processors. If it is a shared memory system, the migration cost will be small as cold miss penalty for task execution. A table in Figure \ref{Figure:MMDF Graph Example} shows the $MC$ value of each task. For instance, $MC$ of node \textit{A} is 30. 
\end{itemize}

We assume that the mode transition of an MMDF graph occurs at the iteration boundary of the SDF schedule associated with the current mode ($cm$). For each task $t \in T$, after the mode transition to the next mode ($nm$), $Rate(p,cm)$ of all $P_t$ and $WCET(t_{cm})$ changes to $Rate(p,nm)$ and $WCET(t_{nm})$. And the SDF schedule associated with the next mode is followed.

With those application and architecture models and profiling information, the problem addressed in this paper is summarized as follows:

\vspace{0.3cm}

\textbf{\underline{PROBLEM}}: Find a mapping and scheduling result of an MMDF graph which satisfies the given throughput constraint

\vspace{0.2cm}

\textbf{\textit{\underline{minimize.}}} the number of required processors

\vspace{0.2cm}

\textbf{\textit{\underline{subject to.}}} the overall throughput performance of the MMDF graph should be higher than the given throughput constraint.

\vspace{0.3cm}

The proposed MMDF scheduling framework is based on a genetic algorithm. So, it needs to evaluate all candidate solutions in every iterations. How to evaluate whether a given mapping and scheduling result of an MMDF graph satisfies the given throughput constraint will be explained in the next section.

\section{Throughput Requirement Analysis}
\label{Section:Throughput Requirement Analysis}

\begin{figure} [ht]
\centerline{\includegraphics[width=1\textwidth]{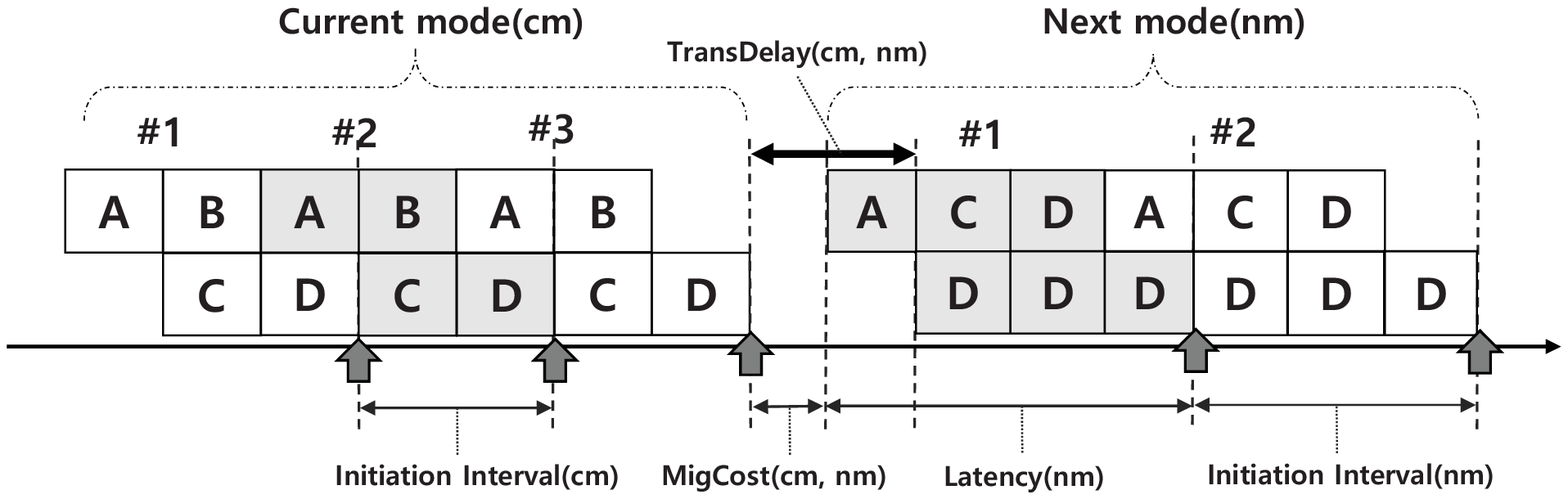}}
\caption{Mode transition delay}
\label{Figure:Mode Transition Delay}
\end{figure}

In this section, we explain how to compute the throughput of an MMDF graph considering the mode transition delay with a given static scheduling results of modes. For simple coordination of task migration and conservative estimation of the mode transition overhead, we assume that mode transition and task migration is performed in a \textit{blocking} fashion. In the \textit{blocking} scheme, the task schedule is \textit{blocked} at the mode transition boundary. Task migration is initiated after all tasks in the current mode finish and the next mode starts after task migration is completed. During mode transitions, no task is executed. Note that the throughput may be degraded due to the blocking even if there is no task migration when a mode transition occurs.

Figure \ref{Figure:Mode Transition Delay} shows a mode transition scenario. It is assumed that a mode changes after the end of iteration 3 in the current mode $cm$. Task $A$ at iteration 1 in next mode $nm$ can start after task $D$ is completed at iteration 3 in mode $cm$ and task migration is completed although a processor is available to execute task $A$ earlier. In the \textit{blocking} scheme of task migration, the mode transition delay consists of two terms: task migration delay and block scheduling effect.

\begin{definition}[Mode Transition Delay]
\label{Definition:Mode Transition Delay}
\begin{equation}
\begin{split}
\forall (cm,nm) \in Trans,\:\: TransDelay(cm, nm) = MigCost(cm, nm) + Latency(nm) \\ - InitiationInterval(nm) \nonumber
\end{split}
\end{equation}
\end{definition}

where $MigCost(cm, nm)$ represents the task migration delay between current mode $cm$ and  next mode $nm$, $Latency(nm)$ denotes the latency between the earliest start time of a task and the latest finish time of a task in mode $nm$, and $InitiationInterval(nm)$ is the start time interval between two consecutive iterations in mode $nm$. Note that $InitiationInterval(nm)$ is equal to the time interval between two consecutive output samples in a streaming application without mode transition. The inverse of the initiation interval denotes the throughput performance if no mode transition occurs.

\subsection{Buffer Size Determination}
\label{SubSection:Buffer Size Determination}

As discussed in Section \ref{SubSection:Mode Transition Delay}, an output buffer is adopted to produce data samples periodically. Since the mode transition delay causes the jitter of output production in an MMDF application, the output buffer should be large enough to provide the data samples during mode transitions. The required output buffer size depends on the maximum mode transition delay and the throughput difference between the input stream and the output stream in the buffer.

To determine the buffer size, we compute the arrival curves of the input and the output streams in the buffer. The arrival curve of a stream informs the number of arriving (or departing) samples (y-axis) within a time interval (x-axis) as shown in Figure \ref{Figure:Arribal Curves} \cite{Thiele:2000}. For conservative estimation, we utilize the maximum arrival curve for the output stream and the minimum arrival curve for the input stream.

\begin{figure} [ht]
\centerline{\includegraphics[width=0.95\textwidth]{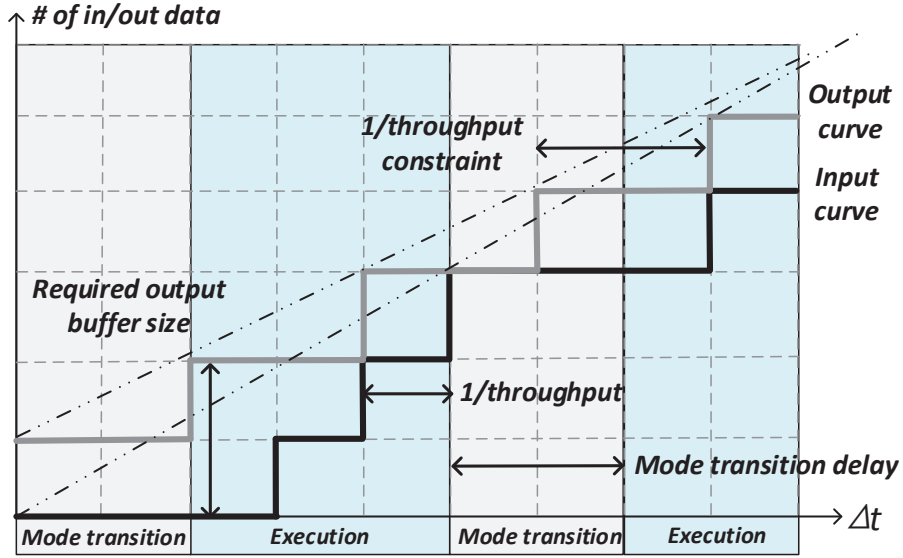}}
\caption{Arrival curves for input and the output streams in the output buffer}
\label{Figure:Arribal Curves}
\end{figure}

In Figure \ref{Figure:Motivational Example of Mode transition} (a), task $Display$ dequeues data from the output buffer periodically with satisfying the throughput constraint, which is depicted as the output curve (gray solid line) in Figure \ref{Figure:Arribal Curves}. The black solid line represents the minimum arrival curve of the input stream which presents the number of generated samples to the output buffer. The buffer size is computed based on the minimum repetition count ($MRC$), the inverse of the throughput, and the maximum mode transition delay among all possible transition scenarios to the mode. In each mode, we compute the buffer size and then choose the maximum buffer size in all modes.

The maximum mode transition delay to mode $m$ is computed as following: 
%is drawn with the maximum mode transition delay among all possible mode transition scenarios in each mode:

\begin{definition}[Worst-case Mode Transition Delay to mode m]
\label{Definition:Worst Case Mode Transition Delay}
%\forall m \in Mode, \:\: 
\begin{gather*}
MaxTransDelay(m) = \max_{\substack{\forall (cm, nm) \\ \in Trans, \\ nm = m}} TransDelay(cm,n m) \:\: 
\end{gather*}
\end{definition}

Note that since the slope of the curve depends on the mode transition delay, the mode repetition count, and the throughput performance of the MMDF schedule, the buffer size is determined after constructing an MMDF schedule meeting the throughput constraints in all modes. 
%in each mode, we have to find the throughput requirement and the corresponding MMDF schedule so as to satisfy the overall throughput constraint.

From the arrival curves, we obtain the minimum output buffer size which is the maximum difference between the curves in every time interval ($\Delta t$). If the overall throughput constraint is satisfied, the output buffer size is computed as following.

\begin{theorem} [Output Buffer Size] The minimum size of the output buffer to satisfy the overall throughput constraint is decided by the following equation:
\label{Theorem:Output Buffer Size}
\begin{gather*}
Output \: buffer \: size = \lceil MaxInterval_{overall} \times ThrConst \rceil
\\ where \:\:MaxInterval_{overall} = \max_{\substack{\forall (cm, nm) \\ \in Trans}} TransDelay(cm, nm) + InitiationInterval(nm) \nonumber
\end{gather*}
\end{theorem}

\begin{proof}
The buffer size is determined by the maximum distance between the input and the output curves, which is illustrated in Figure \ref{Figure:Arribal Curves}. It is evident that the maximum distance between two curves occurs during the first period of the input curve since the tangential slope of the input curve cannot be smaller than that of the output curve. Then the maximum distance is obtained just before the first jump of the input curve. Therefore,
\begin{equation}
\begin{split}
Output \: buffer \: size = \lceil MaxInterval_{overall} \div \frac{1}{ThrConst} \rceil \\ = \lceil MaxInterval_{overall} \times ThrConst \rceil \nonumber
\end{split}
\end{equation}
\end{proof}

\subsection{Throughput Requirement Analysis}
\label{SubSection:Throughput Requirement Analysis}

As discussed above, the throughput requirement at the next iteration depends on the mode transition delay and the minimum repetition counts of the mode. For conservative estimation, the input curve should be steeper than the output curve in all modes. The throughput requirement in each mode can be formulated as follows:

\begin{theorem} [Throughput Requirement] The throughput requirement in mode $m$
%in the MMDF graph, 
denoted as ThrRequire(m), is formulated as following:

\label{Theorem:Throughput Requirement}
\begin{gather*}
ThrRequire(m) = \frac{ThrConst \times MRC(m)}{MRC(m) - (MaxTransDelay(m) \times ThrConst)}  \nonumber
\end{gather*}
\end{theorem}

\begin{proof}
\begin{gather*}
The \: slope \: of \: input \: curve = \frac{MRC(m)}{MaxTransDelay(m) + 1/ThrRequire(m) \times MRC(m)} \\
And \: the \: slope \: of \: output \: curve = \frac{1}{1/ThrConst} = ThrConst \\
Since \: the \: slope \: of \: input \: curve \: should \: be \: steeper \: than \: the \: output \: curve, \\
\frac{MRC(m)}{MaxTransDelay(m) + 1/ThrRequire(m) \times MRC(m)} \geq ThrConst \\
and \:\: ThrRequire(m) \geq \frac{MRC(m)}{MRC(m)/ThrConst - MaxTransDelay(m)}
\end{gather*}
\end{proof}

If the throughput requirement calculated by Theorem \ref{Theorem:Throughput Requirement} is not higher than the throughput in each mode for every MMDF graph, there is a task mapping/scheduling result which satisfies the throughput constraint considering the mode transition delay.

\section{Proposed MMDF Scheduling Framework}
\label{Section:Proposed MMDF Scheduling Framework}

\begin{figure} [ht]
\centerline{\includegraphics[width=0.7\textwidth]{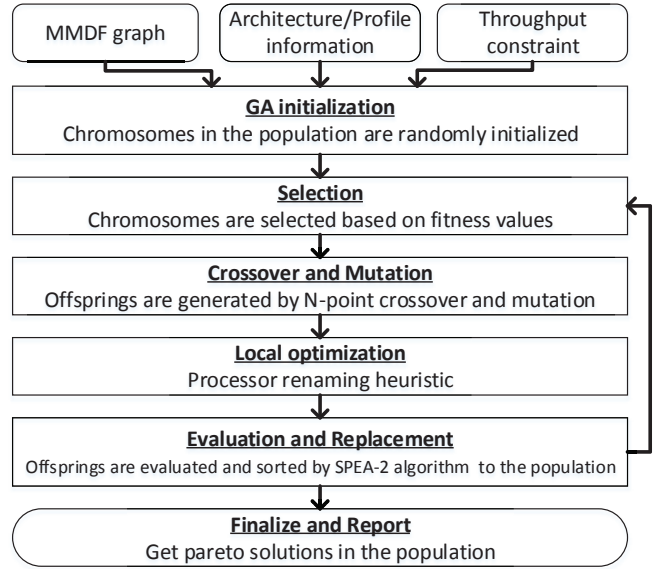}}
\caption{The overall GA framework}
\label{Figure:GA Framework}
\end{figure}

For MMDF scheduling problem, we adopt a genetic algorithm. The overall GA procedure of the proposed framework is shown in Figure \ref{Figure:GA Framework}.

\subsection{GA Configuration}
\label{SubSection:GA Configuration}

\begin{figure} [ht]
\centerline{\includegraphics[width=0.7\textwidth]{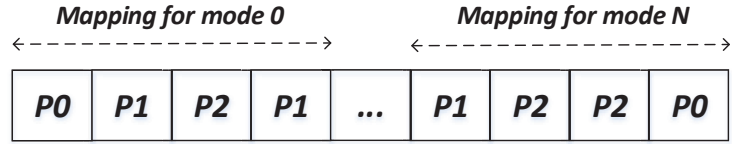}}
\caption{Chromosome structure}
\label{Figure:Chromosome}
\end{figure}

\textbf{\underline{Initialization \& Selection}}: Since a task (or node) can be mapped to different processors in modes, each task is regarded as a unit of mapping in each mode. The chromosome for GA is configured as shown in Figure \ref{Figure:Chromosome}. A chromosome is a set of mapping for each execution mode. Each gene of the chromosome represents to which processor a task in each execution mode is mapped. Chromosomes of initial population are randomly generated and selected for crossover and mutation. The number of selection is a configurable parameter of the GA framework.

\vspace{0.2cm}

\textbf{\underline{Local optimization}}: In order to help the convergence of evolutionary process, a local optimization step is performed before the evaluation step. For local optimization, we devise a processor renaming heuristic that changes the processor id in each mode to reduce the migration cost. The details will be explained later.

\vspace{0.2cm}

\textbf{\underline{Evaluation \& Replacement}}: In this step, we apply a list scheduling heuristic to find a static task schedule in each mode, based on the mapping information given by each chromosome. Once we construct a static schedule, we evaluate the fitness value of each offspring and check whether the throughput constraint is satisfied or not. The fitness function will be described in the next section. Chromosomes in the population are sorted by their fitness values and poor chromosomes are eliminated.

\subsection{Fitness Function}
\label{SubSection:Fitness Function}

The objective of the MMDF scheduling is to minimize the number of processors. The required number of processors is defined as the maximum number of processors in all modes.

\begin{definition}[The Number of Processors for an MMDF Graph]
\label{Definition:The number of processors for an MMDF graph}
\begin{gather*}
The \: number \: of \: processor = \max_{m \in Mode} |Proc_m| \\
where \: Proc_m = \lbrace p \in PE \mid Map(m, p) \neq \emptyset \rbrace
\end{gather*}
\end{definition}

Since the large mode transition delay will degrade the throughput performance and more processors are likely to be required to meet the given throughput constraint, the mode transition delay including task migration overhead is considered to evaluate the number of required processors. And, the GA framework also aims to minimize the overall task migration cost as the secondary objective. The reduction of task migration will save energy consumption of the system, and reduce the network traffic in an NOC architecture. Therefore it is very desirable to reduce the total task migration cost (or delay) in an MMDF graph considering all mode transition scenarios; the total task migration cost of an MMDF graph is defined as follows:

\begin{definition}[Total Task Migration Cost]
\label{Definition:Total task migration cost}
\begin{gather*}
MigCost_{total} = \sum_{(cm, nm) \in Trans} MigCost(cm, nm) \\
where \: MigCost(cm, nm) = \sum_{p \in PE} \sum_{\substack{t \in \lbrace Map(nm, p) \\ - Map(cm, p)\rbrace}} MC(t)
\end{gather*}
\end{definition}

We sum up the migration cost of all possible migration scenarios that are defined by the $MTG$. For each transition in the $MTG$, we accumulate the migration cost of all tasks that are mapped to different processors after the mode transition.

\subsection{Local Optimization Technique}
\label{SubSection:Local Optimization Technique}

\begin{figure} [ht]
\centerline{\includegraphics[width=0.9\textwidth]{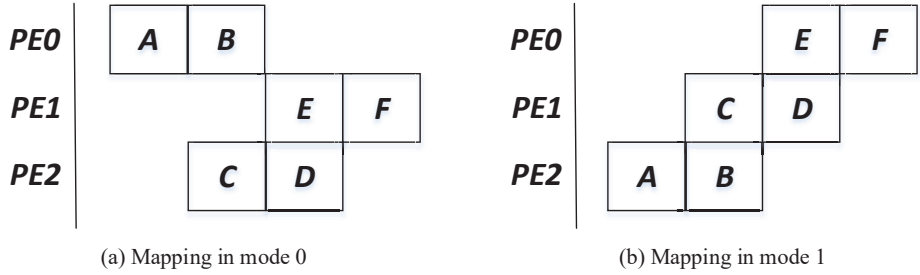}}
\caption{Without processor renaming, every task should be migrated when mode transition occurs. If $PE0$ in mode 0 is renamed to $PE2$ in mode 1, $PE1$ to $PE0$, and $PE2$ to $PE1$, no task migration is required}
\label{Figure:Local Optimization}
\end{figure}

Figure \ref{Figure:Local Optimization} shows a motivational example for local optimization, where two modes have different task mappings defined in the chromosome and a mode transition from mode 0 to mode 1 occurs. In the mapping result, all tasks should be migrated. However, since this paper assumes a homogeneous multiprocessor system, it is possible to rename the processor id in each mode, which is called processor renaming. If $PE0$ in mode 0 is renamed to $PE2$ in mode 1 then tasks A and B do not need to be migrated. Similarly, if $PE1$ in mode 0 is renamed to $PE0$ in mode 1, and $PE2$ to $PE1$ then no task migration is required. Without the processor renaming technique, good solutions such as Figure \ref{Figure:Local Optimization} will be evaluated as poor solutions due to high migration delay, which seriously hinders the convergence of GA.

% Algorithm
\begin{algorithm}[t]
\SetAlgoNoLine
\For{all mode transition scenarios}{
  curr $\leftarrow$ mapping information of src mode of the transition\;
  next $\leftarrow$ mapping information of dst mode of the transition\;
  \For{all mapping information of each processor (cId) in curr}{
    \For{all mapping information of each processor (nId) in next}{
      similarity $\leftarrow$ check similarity between curr[cId] \& next[nId]\;
      \If{(similarity $\ge$ maxSimilarity)}{
        maxSimilarity $\leftarrow$ similarity\;
        swapProcId $\leftarrow$ nId\;
      }
    }
    change mapping of next between cId \& swapProcId\;
  }
}
\caption{Processor Renaming Heuristic}
\label{Algorithm:Processor Renaming Heuristic}
\end{algorithm}

The time complexity of the processor renaming algorithm is given as $P^M$ where $P$ denotes the number of processors and $M$ is the number of mode transition scenarios. Therefore, we devise a simple greedy processor renaming heuristic as shown in Algorithm \ref{Algorithm:Processor Renaming Heuristic} to reduce the time complexity. In the proposed heuristic, the time complexity becomes $O(P^2 \times M)$. Note that processor renaming is only applicable for homogeneous processor systems.

The heuristic measures the similarity between processors. The similarity between processors is defined by how many tasks are mapped on both processors in common.

\begin{definition}[Similarity between processors]
\label{Definition:Similarity between processors}
\begin{gather*}
For (cm, nm) \in Trans, \:\: Similarity(p_i, p_j) = |Map(cm, p_i) \cap Map(nm, p_j)|
\end{gather*}
\end{definition}

For each mode transition, processors in the next mode are renamed to the processors in the current mode with the maximum similarity. Even though the proposed heuristic does not consider all possible processor renaming scenarios and does not provide the optimal renaming result, it reduces the time complexity significantly while generating good quality solutions as confirmed by experimental results.

\section{Experimental Results}
\label{Section:Experimental Result}

\begin{figure} [ht]
\centerline{\includegraphics[width=0.93\textwidth]{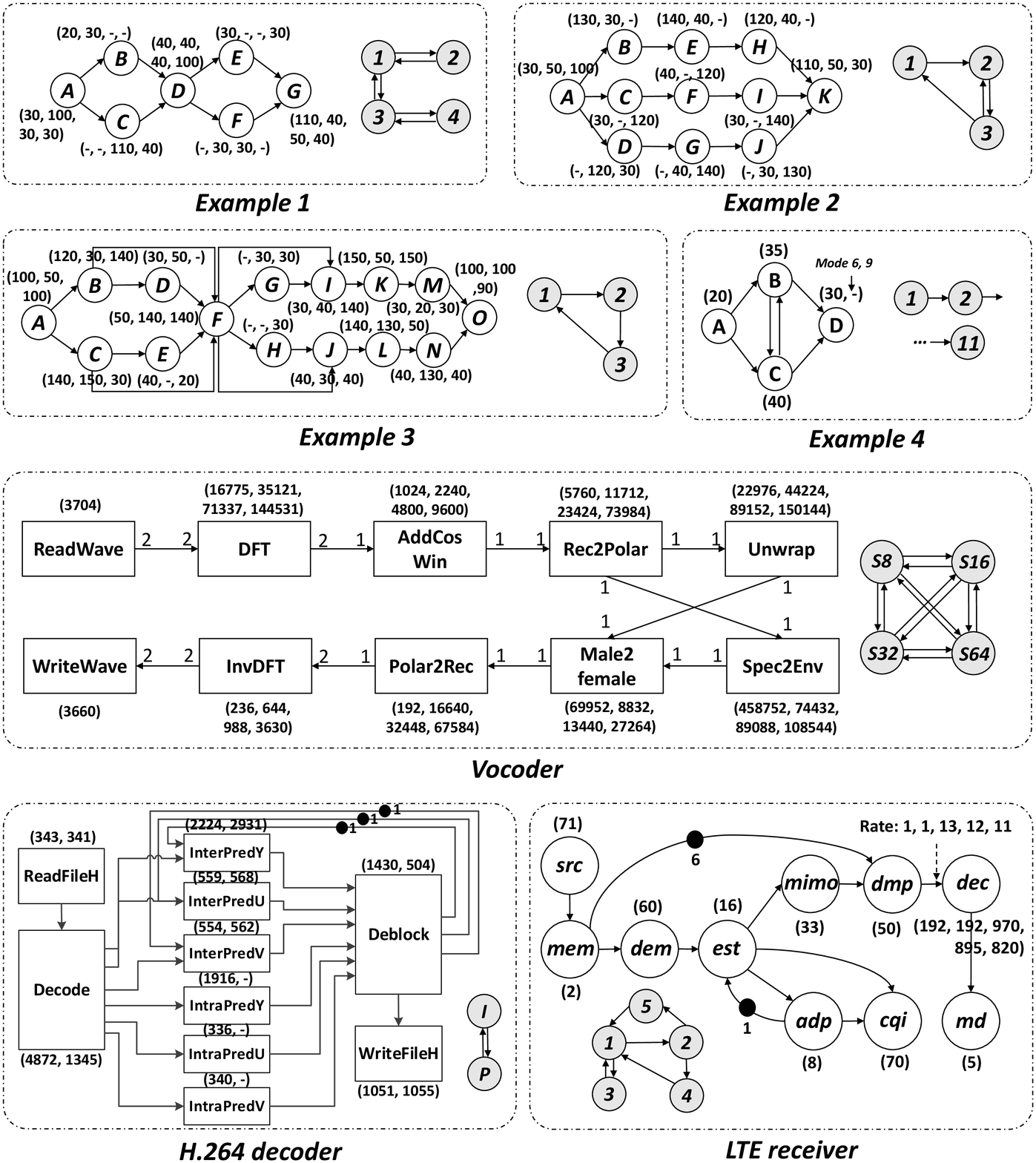}}
\caption{MMDF graph examples used in experiments}
\label{Figure:Experiments Graphs}
\end{figure}

To prove the viability of the proposed framework, we experiment with four synthetic examples and three real applications: H.264 decoder, vocoder \cite{Zhai:2015} and LTE receiver \cite{Siyoum:2011} algorithms as shown in Figure \ref{Figure:Experiments Graphs}. All experiments have been performed on Intel Core i7-4790K 4.00GHz machine with 8GB main memory. Internal parameters of the GA framework are set as shown in Table \ref{Table:GA Configuration}. $\mu$ and $\lambda$ denote the number of parents and offspring.

\begin{table}[ht]
\centering
\tbl{Configuration of the GA framework \label{Table:GA Configuration}}{
\renewcommand{\arraystretch}{1.3}
\small
\begin{tabular}{|
>{\columncolor[HTML]{EFEFEF}}l |l|}
\hline
Population size                     & 100   \\ \hline
$\mu$ and $\lambda$                 & 100   \\ \hline
Probabilities of crossover/mutation & 0.9   \\ \hline
Maximum generations                 & 30000 \\ \hline
\end{tabular}}
\end{table}

\subsection{MMDF Scheduling Technique}
\label{SubSection:MMDF Scheduling Technique}

We compare the proposed technique with two different heuristics. The first heuristic schedules SDF graphs independently and performs the processor renaming heuristic. We denote this technique as \textit{Base}. The \textit{Base} technique is an iterative algorithm. First, for each mode, it constructs pareto-optimal solutions which are optimized with throughput and the number of processors using a genetic algorithm. Then it selects a initial schedule which satisfies the throughput constraint with the minimum number of processors for each mode. Based on the mapping/scheduling results, it performs the processor renaming heuristic and adjusts the throughput requirement as discussed in the previous section, considering the mode transition delay incurred by the initial schedules. If a schedule does not satisfy the calculated throughput requirement, it is replaced with another schedule which uses one more processor. Unless all scheduling results satisfy the new adjusted throughput requirement in all modes, it repeats the mapping/scheduling with the new adjusted throughput requirement until the mapping/scheduling results satisfy the adjusted throughput requirement.

The second approach fixes task mapping in all modes disallowing task migration as the existing approaches usually assume. This technique is denoted as \textit{Fixed}. The \textit{Fixed} technique is implemented in the same GA framework as the proposed framework with disallowing task migration only.

Each MMDF graph consists of a task graph and the associated \textit{MTG}. For all graphs in Figure \ref{Figure:Experiments Graphs} exclude the vocoder application, $Rate(p, mode)$ for each port $p \in P_t$ is one except the input port of node \textit{dec} in the LTE receiver application. For the vocoder application, port rates are fixed among all modes and specified in the figure. For the task graph of vocoder application in \cite{Zhai:2015}, we reduce the number of invocations for specific tasks (from \textit{AddCosWin} to \textit{Polar2Rec}) from 128 to 2 by clustering, so the given WCETs of those tasks in \cite{Zhai:2015} are multiplied by 64. Also, we allow that each instance of the same node can be mapped onto different processors for data parallelism.

The numbers above or under the tasks in Figure \ref{Figure:Experiments Graphs} indicate the $WCET(t_m)$ in each mode. In case that the $WCET(t_m)$ of a task is constant in all modes, a single number is denoted. For synthetic examples, the WCET of each task is set to an arbitrary value, and the WCET of each task in the H.264 decoder application is set to profiled data with $us$ unit. Also, the WCET of each task in the vocoder and LTE receiver applications is set to a value given in \cite{Zhai:2015} and \cite{Siyoum:2011}.

\begin{table}[ht]
\centering
\tbl{Configurations for experiments \label{Table:Graph Configuration}}{
\renewcommand{\arraystretch}{1.3}
\small
\begin{tabular}{|
>{\columncolor[HTML]{EFEFEF}}  c| x{4.5cm} | x{4.5cm} |}
\hline
               & \cellcolor[HTML]{EFEFEF} $MRC(m)$                       & \cellcolor[HTML]{EFEFEF} $ThrConst$                       \\ \hline
\textit{Example 1}      & $\forall m \in Mode,\:MRC(m) = 5$              & 1/150 iteration/time-unit                                 \\ \hline
\textit{Example 2}      & $\forall m \in Mode,\:MRC(m) = 5$              & 1/260 iteration/time-unit                                 \\ \hline
\textit{Example 3}      & $\forall m \in Mode,\:MRC(m) = 5$              & 1/330 iteration/time-unit                                 \\ \hline
\textit{Example 4}      & $\forall m \in Mode,\:MRC(m) = 5$              & 1/80 iteration/time-unit                                  \\ \hline
\textit{Vocoder}        & $\forall m \in Mode,\:MRC(m) = 5$              & 1/500000 iteration/cycle                                  \\ \hline
\textit{H.264 decoder}  & $MRC(I) = 1, MRC(P) = 5$                       & 1/12500 iteration/us (80 fps)                             \\ \hline
\textit{LTE receiver}   & $\forall m \in Mode,\:MRC(m) = 5$              & 1/1800 sub-frames/time-unit                               \\ \hline
\end{tabular}}
\end{table}

For all configurations in Table \ref{Table:Graph Configuration}, we compare the following three techniques: \textit{Base}, \textit{Fixed}, and \textit{Proposed}. We assume that the minimum repetition count ($MRC$) for all modes in each example is set to 5 except the H.264 decoder application, since the mode transition pattern of the H.264 decoder is known and fixed (eg. \textit{I-P-P-P-P-P-I-P-P-P-...}). Throughput constraints are set arbitrarily with considering the WCET of tasks.

\begin{figure} [ht]
\centerline{\includegraphics[width=1\textwidth]{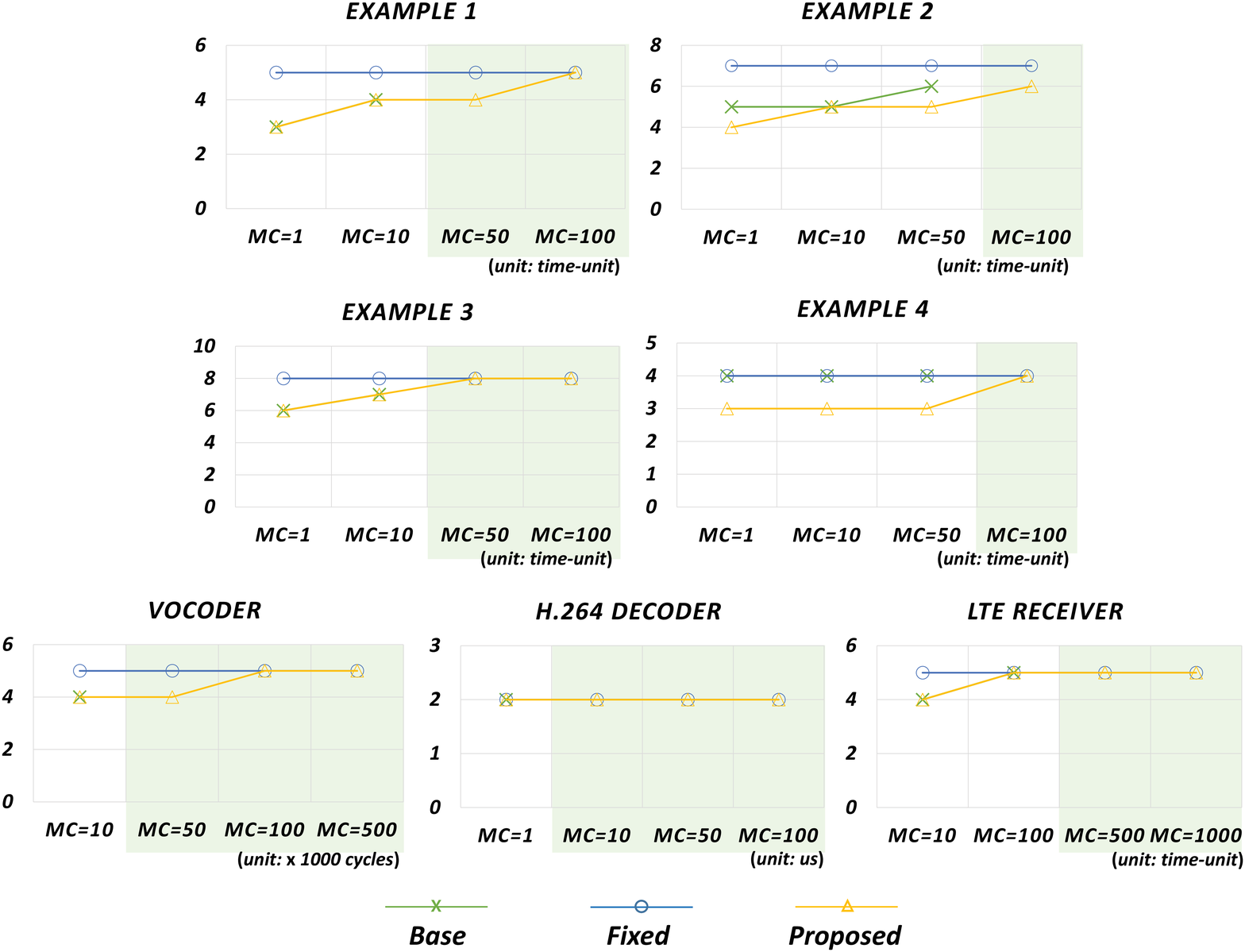}}
\caption{Comparison results in terms of the number of processors: \textit{Base}, \textit{Fixed} and \textit{Proposed}}
\label{Figure:Comparison Result}
\end{figure}

Figure \ref{Figure:Comparison Result} shows the experimental results for all applications. The y-axis indicates the number of required processors and the x-axis presents migration cost ($MC(t)$). To show the viability of the proposed framework, we examine how the mode transition delay takes effect on the number of processors. For each example, we vary the task migration cost or $MC(t)$ using four different values. In the synthetic examples and the vocoder/LTE receiver applications, the task migration cost is fixed as $MC(t)$ for all tasks. In H.264 decoder applications, however, $MC(t)$ is scaled based on the actual task code size for all $t \in T$: the task migration cost of a task is computed as the product of $MC(t)$ values in x-axis and its task code size.

The results show that the \textit{Proposed} method requires no more processors than \textit{Base} and \textit{Fixed} approaches. The \textit{Fixed} approach requires more processors in most cases than \textit{Base} and \textit{Proposed} approaches. Since the \textit{Fixed} approach does not allow task migration, the number of required processors to meet the given throughput constraint is independent of the task migration cost. Since the mode transition delay is determined by not only the task migration delay but also the latency in the \textit{blocking} scheme of task migration, when $MC(t)$ is small, the mode transition delay is mostly dependent on the latency. Therefore, the \textit{Base} approach shows the similar results to the \textit{Proposed} approach for small migration costs. However, as $MC(t)$ increases, the \textit{Base} approach requires more processors than the \textit{Proposed} approach or could not find a feasible solution for large $MC(t)$ values in cases which are highlighted with a green box in Figure \ref{Figure:Comparison Result}.

In the H.264 decoder application, there exists a dominant mode in which all tasks in an MMDF graph are executed. Since the dominant mode creates a critical path in all modes, if a mapping and scheduling result satisfies the throughput constraint in the dominant mode then results in the other modes automatically satisfy the throughput constraint. Hence, no task migration is required and \textit{Fixed} and \textit{Proposed} approaches produce the same results for the application.

\begin{table}[ht]
\centering
\tbl{Experimental results of \textit{Proposed} in case of $MC(t)$ = 10 \label{Table:Detail Result}}{
\renewcommand{\arraystretch}{1.3}
\small
\begin{tabular}{|
>{\columncolor[HTML]{EFEFEF}}c |c|c|c|c|}
\hline
                        & \cellcolor[HTML]{EFEFEF} \makecell{$\max_{m \in Mode}$ \\ $MaxTransDelay(m)$} & \cellcolor[HTML]{EFEFEF} \makecell{$\min_{m \in Mode}$ \\ $ThrRequire(m)$} & \cellcolor[HTML]{EFEFEF} \makecell{$Output$ \\ $buffer\:size$} & \cellcolor[HTML]{EFEFEF} $MigCost_{total}$ \\ \hline
\textit{Example 1}      & 150 time-unit                      & 1/120                            & 2                             & 20 time-unit                                         \\ \hline
\textit{Example 2}      & 380 time-unit                      & 1/184                            & 3                             & 20 time-unit                                         \\ \hline
\textit{Example 3}      & 640 time-unit                      & 1/202                            & 3                             & 60 time-unit                                         \\ \hline
\textit{Example 4}      & 70 time-unit                       & 1/66                             & 2                             & 50 time-unit                                         \\ \hline
\textit{Vocoder}        & 280373 cycles                      & 1/443925                         & 2                             & 150000 cycles                                        \\ \hline
\textit{H.264 decoder}  & 746 us                             & 1/12350.8                        & 1                             & 0 time-unit                                          \\ \hline
\textit{LTE receiver}   & 980 time-unit                      & 1/1604                           & 2                             & 30 time-unit                                         \\ \hline
\end{tabular}}
\end{table}

Table \ref{Table:Detail Result} presents the detailed experimental results from the \textit{Proposed} approach in Figure \ref{Figure:Comparison Result} when $MC(t) = 10$ for instance. The table informs that the throughput which an application should satisfy becomes tighter than the given throughput constraint in Table \ref{Table:Graph Configuration} due to the mode transition delay. The table also presents the total task migration cost and the required output buffer sizes for benchmark applications.

\subsection{Scalability of the Proposed Framework}
\label{Scalability of the Proposed Framework}

\begin{figure} [ht]
\centerline{\includegraphics[width=0.7\textwidth]{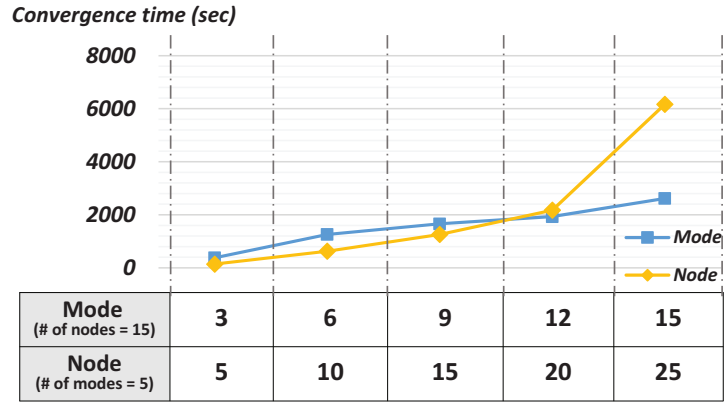}}
\caption{Experimental result for the scalability property}
\label{Figure:Scalability Result}
\end{figure}

Because the proposed framework is based on the multi-objective genetic algorithm, its convergence speed depends on the size of solution space. As shown in Figure \ref{Figure:Chromosome}, the size of the solution space depends on the number of nodes and modes. So, we perform experiments for different configurations of these factors. Figure \ref{Figure:Scalability Result} shows the experimental results of the scalability of the proposed framework for synthetic examples. The results show that the number of nodes more contributes to the convergence speed than the number of modes.

\section{Conclusion}
\label{Section:Conclusion}

In this paper, we address the multiprocessor scheduling problem of a multi-mode dataflow (MMDF) graph allowing task migration with non-negligible mode transition delay. An MMDF graph has a finite set of modes and each mode is specified by an SDF graph. We observe that the mode transition delay should be considered in many streaming applications in which the mode transition occurs frequently, in order to satisfy the throughput constraint. Thus we propose a mapping/scheduling framework based on a genetic algorithm which schedules all SDF graphs simultaneously to minimize the number of processors while keeping the throughput constraint. Also, we propose the formulations to compute the required buffer size and the required throughput performance of the MMDF graph to satisfy the given throughput constraint of the system, by estimating the mode transition delay conservatively from an obtained scheduling result under the assumption of \textit{blocking} scheme of task migration.
To show the viability of the proposed technique, we compare the proposed technique with two other approaches with some synthetic examples and three real applications. Experimental results confirm the superiority of the proposed technique over other approaches.

Because the mode transition delay is conservatively calculated with \textit{blocking} scheme of task migration, as a future work, we plan to calculate the throughput requirement with more exact model of mode transition.

% Acknowledgments
\begin{acks}
This research was supported by a grant to Bio-Mimetic Robot Research Center Funded by Defense Acquisition Program Administration, and by Agency for Defense Development (UD130070ID), Basic Science Research Program through the National Research Foundation of Korea (NRF) funded by the Ministry of Science, ICT \& Future Planning (NRF-2013R1A2A2A01067907, 2013R1A1A1013384), and IT R\&D program MKE/KEIT (No. 10041608, Embedded system Software for New-memory based Smart Device).
\end{acks}

% Bibliography
\bibliographystyle{ACM-Reference-Format-Journals}
\bibliography{IDEA_2016_jhw-bibfile.bib}


\begin{thebibliography}{00}

%%% the Web address.

\ifx \showCODEN    \undefined \def \showCODEN     #1{\unskip}     \fi
\ifx \showDOI      \undefined \def \showDOI       #1{{\tt DOI:}\penalty0{#1}\ }
  \fi
\ifx \showISBNx    \undefined \def \showISBNx     #1{\unskip}     \fi
\ifx \showISBNxiii \undefined \def \showISBNxiii  #1{\unskip}     \fi
\ifx \showISSN     \undefined \def \showISSN      #1{\unskip}     \fi
\ifx \showLCCN     \undefined \def \showLCCN      #1{\unskip}     \fi
\ifx \shownote     \undefined \def \shownote      #1{#1}          \fi
\ifx \showarticletitle \undefined \def \showarticletitle #1{#1}   \fi
\ifx \showURL      \undefined \def \showURL       #1{#1}          \fi

\bibitem[\protect\citeauthoryear{Bebelis, Fradet, Girault, and
  Lavigueur}{Bebelis et~al\mbox{.}}{2013}]%
        {Bebelis:2013}
{Vagelis Bebelis}, {Pascal Fradet}, {Alain Girault}, {and} {Bruno Lavigueur}.
  2013.
\newblock \showarticletitle{BPDF: A Statically Analyzable DataFlow Model with
  Integer and Boolean Parameters}. In {\em Proceedings of the Eleventh ACM
  International Conference on Embedded Software} {\em (EMSOFT '13)}. IEEE
  Press, Piscataway, NJ, USA, Article 3, 10 pages.
\newblock
\showISBNx{978-1-4799-1443-2}
\showURL{%
\url{http://dl.acm.org/citation.cfm?id=2555754.2555757}}


\bibitem[\protect\citeauthoryear{Bhattacharya and Bhattacharyya}{Bhattacharya
  and Bhattacharyya}{2001}]%
        {Bhattacharya:2001}
{Bishnupriya Bhattacharya} {and} {Shuvra~S. Bhattacharyya}. 2001.
\newblock \showarticletitle{Parameterized dataflow modeling for DSP systems}.
\newblock {\em Signal Processing, IEEE Transactions on\/} {49}, 10 (Oct 2001),
  2408--2421.
\newblock
\showISSN{1053-587X}
\showDOI{%
\url{http://dx.doi.org/10.1109/78.950795}}


\bibitem[\protect\citeauthoryear{Bilsen, Engels, Lauwereins, and
  Peperstraete}{Bilsen et~al\mbox{.}}{1995}]%
        {Bilsen:1995}
{Greet Bilsen}, {Marc Engels}, {Rudy Lauwereins}, {and} {J.A. Peperstraete}.
  1995.
\newblock \showarticletitle{Cyclo-static data flow}. In {\em Acoustics, Speech,
  and Signal Processing, 1995. ICASSP-95., 1995 International Conference on},
  Vol.~5. 3255--3258 vol.5.
\newblock
\showISSN{1520-6149}
\showDOI{%
\url{http://dx.doi.org/10.1109/ICASSP.1995.479579}}


\bibitem[\protect\citeauthoryear{Damavandpeyma, Stuijk, Basten, Geilen, and
  Corporaal}{Damavandpeyma et~al\mbox{.}}{2013}]%
        {Damavandpeyma:2013}
{Morteza Damavandpeyma}, {Sander Stuijk}, {Twan Basten}, {Marc Geilen}, {and}
  {Henk Corporaal}. 2013.
\newblock \showarticletitle{Throughput-constrained DVFS for scenario-aware
  dataflow graphs}. In {\em Real-Time and Embedded Technology and Applications
  Symposium (RTAS), 2013 IEEE 19th}. 175--184.
\newblock
\showISSN{1080-1812}
\showDOI{%
\url{http://dx.doi.org/10.1109/RTAS.2013.6531090}}


\bibitem[\protect\citeauthoryear{Geilen and Stuijk}{Geilen and Stuijk}{2010}]%
        {Geilen:2010}
{Marc Geilen} {and} {Sander Stuijk}. 2010.
\newblock \showarticletitle{Worst-case Performance Analysis of Synchronous
  Dataflow Scenarios}. In {\em Proceedings of the Eighth IEEE/ACM/IFIP
  International Conference on Hardware/Software Codesign and System Synthesis}
  {\em (CODES/ISSS '10)}. ACM, New York, NY, USA, 125--134.
\newblock
\showISBNx{978-1-60558-905-3}
\showDOI{%
\url{http://dx.doi.org/10.1145/1878961.1878985}}


\bibitem[\protect\citeauthoryear{Geilen, Stuijk, and Basten}{Geilen
  et~al\mbox{.}}{2012}]%
        {Geilen:2012}
{Marc Geilen}, {Sander Stuijk}, {and} {Twan Basten}. 2012.
\newblock \showarticletitle{Predictable dynamic embedded data processing}. In
  {\em Embedded Computer Systems (SAMOS), 2012 International Conference on}.
  320--327.
\newblock
\showDOI{%
\url{http://dx.doi.org/10.1109/SAMOS.2012.6404194}}


\bibitem[\protect\citeauthoryear{Girault, Lee, and Lee}{Girault
  et~al\mbox{.}}{1999}]%
        {Girault:1999}
{Alain Girault}, {Bilung Lee}, {and} {Edward~A. Lee}. 1999.
\newblock \showarticletitle{Hierarchical finite state machines with multiple
  concurrency models}.
\newblock {\em Computer-Aided Design of Integrated Circuits and Systems, IEEE
  Transactions on\/} {18}, 6 (Jun 1999), 742--760.
\newblock
\showISSN{0278-0070}
\showDOI{%
\url{http://dx.doi.org/10.1109/43.766725}}


\bibitem[\protect\citeauthoryear{Jung, Lee, haeng Kang, Kim, Oh, and Ha}{Jung
  et~al\mbox{.}}{2014}]%
        {Jung:2014}
{Hanwoong Jung}, {Chanhee Lee}, {Shin haeng Kang}, {Sungchan Kim}, {Hyunok Oh},
  {and} {Soonhoi Ha}. 2014.
\newblock \showarticletitle{Dynamic Behavior Specification and Dynamic Mapping
  for Real-Time Embedded Systems: HOPES Approach}.
\newblock {\em ACM Trans. Embed. Comput. Syst.\/} {13}, 4s, Article 135 (April
  2014), 26 pages.
\newblock
\showISSN{1539-9087}
\showDOI{%
\url{http://dx.doi.org/10.1145/2584658}}


\bibitem[\protect\citeauthoryear{Lee, Kim, Oh, and Ha}{Lee
  et~al\mbox{.}}{2013}]%
        {Lee:2013}
{Chanhee Lee}, {Sungchan Kim}, {Hyunok Oh}, {and} {Soonhoi Ha}. 2013.
\newblock \showarticletitle{Failure-Aware Task Scheduling of Synchronous Data
  Flow Graphs Under Real-Time Constraints}.
\newblock {\em Journal of Signal Processing Systems\/} {73}, 2 (2013),
  201--212.
\newblock
\showISSN{1939-8018}
\showDOI{%
\url{http://dx.doi.org/10.1007/s11265-013-0753-3}}


\bibitem[\protect\citeauthoryear{Lee and Messerschmitt}{Lee and
  Messerschmitt}{1987}]%
        {Lee:1987}
{Edward~A. Lee} {and} {David~G. Messerschmitt}. 1987.
\newblock \showarticletitle{Synchronous data flow}.
\newblock {\it Proc. IEEE} {75}, 9 (Sept 1987), 1235--1245.
\newblock
\showISSN{0018-9219}
\showDOI{%
\url{http://dx.doi.org/10.1109/PROC.1987.13876}}


\bibitem[\protect\citeauthoryear{Moreira}{Moreira}{2012}]%
        {Moreira:2012}
{Orlando Moreira}. 2012.
\newblock \showarticletitle{Temporal analysis and scheduling of hard real-time
  radios running on a multi-processor}.
\newblock {\em ser. PHD Thesis, Technische Universiteit Eindhoven\/} (2012).
\newblock


\bibitem[\protect\citeauthoryear{Siyoum, Geilen, Moreira, Nas, and
  Corporaal}{Siyoum et~al\mbox{.}}{2011}]%
        {Siyoum:2011}
{Firew Siyoum}, {Marc Geilen}, {Orlando Moreira}, {Rick Nas}, {and} {Henk
  Corporaal}. 2011.
\newblock \showarticletitle{Analyzing synchronous dataflow scenarios for
  dynamic software-defined radio applications}. In {\em System on Chip (SoC),
  2011 International Symposium on}. 14--21.
\newblock
\showDOI{%
\url{http://dx.doi.org/10.1109/ISSOC.2011.6089222}}


\bibitem[\protect\citeauthoryear{Stuijk, Geilen, and Basten}{Stuijk
  et~al\mbox{.}}{2010}]%
        {Stuijk:2010}
{Sander Stuijk}, {Marc Geilen}, {and} {Twan Basten}. 2010.
\newblock \showarticletitle{A Predictable Multiprocessor Design Flow for
  Streaming Applications with Dynamic Behaviour}. In {\em Digital System
  Design: Architectures, Methods and Tools (DSD), 2010 13th Euromicro
  Conference on}. 548--555.
\newblock
\showDOI{%
\url{http://dx.doi.org/10.1109/DSD.2010.31}}


\bibitem[\protect\citeauthoryear{Stuijk, Geilen, Theelen, and Basten}{Stuijk
  et~al\mbox{.}}{2011}]%
        {Stuijk:2011}
{Sander Stuijk}, {Marc Geilen}, {Bart~D. Theelen}, {and} {Twan Basten}. 2011.
\newblock \showarticletitle{Scenario-aware dataflow: Modeling, analysis and
  implementation of dynamic applications}. In {\em Embedded Computer Systems
  (SAMOS), 2011 International Conference on}. 404--411.
\newblock
\showDOI{%
\url{http://dx.doi.org/10.1109/SAMOS.2011.6045491}}


\bibitem[\protect\citeauthoryear{Stuijk, Ghamarian, Theelen, Geilen, and
  Basten}{Stuijk et~al\mbox{.}}{2008}]%
        {Stuijk:2008}
{Sander Stuijk}, {Amirhossein Ghamarian}, {Bart~D. Theelen}, {Marc Geilen},
  {and} {Twan Basten}. 2008.
\newblock {\em FSM-based SADF}.
\newblock {T}echnical {R}eport. Citeseer.
\newblock


\bibitem[\protect\citeauthoryear{Thiele, Chakrabort, and Naedele}{Thiele
  et~al\mbox{.}}{2000}]%
        {Thiele:2000}
{Lothar Thiele}, {Samarjit Chakrabort}, {and} {Martin Naedele}. 2000.
\newblock \showarticletitle{Real-time calculus for scheduling hard real-time
  systems}. In {\em Circuits and Systems, 2000. Proceedings. ISCAS 2000 Geneva.
  The 2000 IEEE International Symposium on}, Vol.~4. 101--104 vol.4.
\newblock
\showDOI{%
\url{http://dx.doi.org/10.1109/ISCAS.2000.858698}}


\bibitem[\protect\citeauthoryear{Wiggers, Bekooij, and Smit}{Wiggers
  et~al\mbox{.}}{2008}]%
        {Wiggers:2008}
{Maarten~H. Wiggers}, {Marco J.~G. Bekooij}, {and} {Gerard J.~M. Smit}. 2008.
\newblock \showarticletitle{Buffer Capacity Computation for Throughput
  Constrained Streaming Applications with Data-Dependent Inter-Task
  Communication}. In {\em Real-Time and Embedded Technology and Applications
  Symposium, 2008. RTAS '08. IEEE}. 183--194.
\newblock
\showISSN{1545-3421}
\showDOI{%
\url{http://dx.doi.org/10.1109/RTAS.2008.10}}


\bibitem[\protect\citeauthoryear{Zhai}{Zhai}{2015}]%
        {Zhai:2015}
{Jiali~Teddy Zhai}. 2015.
\newblock \showarticletitle{Adaptive streaming applications : analysis and
  implementation models}.
\newblock {\em ser. PHD Thesis, Universiteit Leiden\/} (2015).
\newblock


\end{thebibliography}
                             % Sample .bib file with references that match those in
                             % the 'Specifications Document (V1.5)' as well containing
                             % 'legacy' bibs and bibs with 'alternate codings'.
                             % Gerry Murray - March 2012

\medskip

\end{document}